\documentclass[reqno]{amsart}
\usepackage[colorlinks=true,linkcolor=blue]{hyperref}    
\usepackage{amssymb}
\usepackage{amsfonts}
\usepackage{amsthm}
\usepackage{amsmath}
\usepackage{latexsym}
\usepackage{color}

\numberwithin{equation}{section}
\newcommand{\point}{\par\noindent$\bullet$ \ }

\newcommand{\R}{{\mathbb R}}

\newcommand{\Z}{{\mathbb Z}}
\newcommand{\N}{{\mathbb N}}
\newcommand{\Q}{{\mathbb Q}}
\newcommand{\C}{{\mathbb C}}
\newcommand{\Y}{\Upsilon}
\renewcommand{\mod}{{\rm mod}\,}
\newcommand{\re}{{\rm Re}\,}
\newcommand{\im}{{\rm Im}\,}

\newcommand{\res}{{\rm res}\, }
\newcommand{\Const}{{\rm Const}}
\newcommand{\F}{{\mathcal F}}
\renewcommand{\P}{{\mathcal P}}
\theoremstyle{plain}
\newtheorem{Th}{Theorem}[section]
\newtheorem{Le}{Lemma}[section]
\newtheorem{Pro}{Proposition}[section]

\theoremstyle{definition}
\newtheorem{Rem}{Remark}[section]
\newtheorem{Def}{Definition}[section]
\begin{document}
%
\title{An exact renormalization formula for the Maryland model}
\thanks{
\indent \ $^1$ Department of Mathematical Physics, St.~Petersburg State University,  Ulianovskaja, 1,  St.~Petersburg-Petrodvoretz, 198904 Russia, E-mail: \href{mailto:fedotov.s@mail.ru}{fedotov.s@mail.ru}\\
\indent \ $^2$ Chebyshev Laboratory, St.~Petersburg State University,
14th Line, 29b, Vasilyevsky Island, St.~Petersburg, 199178 Russia, E-mail: \href{mailto:sandomirski@yandex.ru}{sandomirski@yandex.ru}\\
\indent \ The work of the first author is supported by the Russian Foundation of
Basic Research under grant 11-01-00458-a; 
the work of the second one is supported by the Chebyshev 
Laboratory (St. Petersburg State University)  under RF Government 
grant 11.G34.31.0026 and by JSC ``Gazprom Neft''}
\maketitle
\vskip -0.5cm
\begin{center}
{\small{ ALEXANDER FEDOTOV$^{1}$ AND FEDOR SANDOMIRSKIY}}$^{2,1}$
\end{center}


%
%
\begin{abstract}
We discuss the difference Schr\"o\-dinger equation
$\psi_{k+1}+\psi_{k-1}+\lambda \cot(\pi\omega k+\theta)\psi_k=E\psi_k$, $k\in\Z$,
where $\lambda$, $\omega$, $\theta$ and $E$ are parameters.
We obtain explicit renormalization formulas relating its solutions 
for large $|k|$  to solutions of the equation with new parameters $\lambda$, 
$\omega$, $\theta$ and $E$ for bounded $|k|$. These formulas are similar to the 
renormalization formulas from the theory of Gaussian  exponential sums.
\end{abstract}
\vskip 1cm
\section{Introduction}
We consider  the difference Schr\"odinger equation
\begin{equation}
\label{eq_maryland_lattice_equation}
\psi_{k+1}+\psi_{k-1}+\lambda\cot(\pi(\omega k+\theta))\psi_k=E\psi_k,
\quad k\in\Z,
\end{equation}
where $\omega\in(0,1)\setminus\Q$, \ $\theta\in[0,1)$, \ $\lambda>0$   
and $E\in\R$ are parameters; $E$ is called the {\it spectral parameter}.

The Schr\"odinger operator in $l^2(\Z)$ corresponding 
to~\eqref{eq_maryland_lattice_equation} is referred to as the Maryland model.
It is one of the popular models of spectral 
theory~\cite{CFKS,PF}:  being a non-trivial almost periodic operator, 
many of its important spectral properties can be explicitly described. 
There are interesting open problems related to the behavior of solutions 
of~\eqref{eq_maryland_lattice_equation} for large $|k|$. For example, 
one can mention the study of the spectrum of the Maryland model for 
frequencies that are neither well no badly approximable by rational 
numbers, e.g.~\cite{PF}, the investigation of the multiscale behavior of 
its (generalized) eigenfunctions, and the explanation of the time evolution 
generated by the Maryland model, e.g.~\cite{FGP}.

In this paper, for sake of brevity we call~\eqref{eq_maryland_lattice_equation}
the Maryland equation.

The central result of this paper is a renormalization formula expressing 
solutions of~\eqref{eq_maryland_lattice_equation} in terms of solutions of 
the Maryland equation with new parameters $\omega,\theta,\lambda, E$ 
for smaller $|k|$. This formula is similar to the well-known renormalization
formula from the theory of Gaussian exponential sums, see, for 
example,~\cite{FK:12}. 

To describe the main result, define the parameters $l>0$ and $-\pi<\eta<\pi$ 
in  terms of  $E$ and $\lambda$ so that  $E+i\lambda=2\cos (\eta+il)$.
Then
\begin{equation}
\label{eq_Gamma_l_E_lambda}
\lambda=-2{\rm sh}\, l \sin\eta,\qquad E=2{\rm ch}\, l \cos\eta.
\end{equation}
Put
\begin{equation}\label{eq:matrix-F}
\F(z,\eta,l)=\begin{pmatrix}
2{\rm ch}\, l \cos\eta +2{\rm sh}\, l \sin\eta\cot (\pi z) & -1\\1 & 0
\end{pmatrix}.
\end{equation}
The Maryland equation~\eqref{eq_maryland_lattice_equation} 
is equivalent to the equation
\begin{equation}\label{matrix-Mary}
  \Psi_{k+1}=\F(k\omega+\theta,\eta,l)\Psi_k,\quad k\in\Z
\end{equation}
(write down the equation for the first component of a vector solution 
of~\eqref{matrix-Mary} !). Let $\P_k(\omega,\theta,\eta,l)$ be the 
matrix solution of~(\ref{matrix-Mary}) that is equal to the identity 
matrix for $k=0$. Note that
\begin{gather*}
  \P_k(\omega,\theta,\eta,l)=\F(\theta+(k-1)\,\omega,\eta,l)\dots
\F(\theta+\omega,\eta,l)\,\F(\theta,\eta,l),\quad k\ge1,\\
\P_k(\omega,\theta,\eta,l)=\F(\theta+k\,\omega,\eta,l)^{-1}\dots
\F(\theta-2\omega,\eta,l)^{-1}\,\F(\theta-\omega,\eta,l)^{-1},\quad k\le-1.
\end{gather*}
The main result is described by
\begin{Th}\label{th:main} For any $N\in\Z$, one has
\begin{equation}
\label{eq:main_renormalization_formulae}
\P_N(\omega,\theta,\eta,l)=
\Psi(\{\theta+N\omega\},\eta,l)\,
\sigma_2\,\P_{N_1}(\omega_1,\theta_1,\eta_1,l_1)\,\sigma_2\,\Psi^{-1}(\theta,\eta,l),
\end{equation}
where 
\begin{equation}
\label{eq_new_parameters_and_old}
N_1=-[\theta+N\omega],\quad
\omega_1=\left\{\frac{1}{\omega}\right\},\quad 
\theta_1=\left\{\frac{\theta}{\omega}\right\},\quad
\eta_1=\frac{\eta}\omega\mod2\pi,\quad l_1=\frac{l}\omega,
\end{equation}
$[\,x\,]$ and  $\{\,x\,\}$ denote the integer and the fractional parts of $x\in\R$,
\begin{equation}\label{eq:Psi-sigma}
\Psi(z,\eta,l)=\begin{pmatrix} \psi(z,\eta,l) & \psi(z-1,\eta,l)\\
\psi(z-\omega,\eta,l) & \psi(z-1-\omega,\eta,l)\end{pmatrix},\quad
\sigma_2=\begin{pmatrix} 0 & -i\\ i& 0\end{pmatrix},
\end{equation}
and $\psi$ is the minimal meromorphic solution of the ``complex Maryland equation''
\begin{equation}
  \label{eq:complex-Mary}
  \psi(z+\omega)+\psi(z-\omega)+\lambda\cot(\pi z)\psi(z)=E\,\psi(z),\quad z\in \C.
\end{equation}
\end{Th}
The importance of the minimal entire solutions, i.e., the solutions 
having the slowest possible growth for $\im z\to\pm \infty$,  for the 
study of difference equations with entire periodic coefficients was 
revealed in~\cite{BF:01}. For equation~\eqref{eq:complex-Mary}, the 
definition of the minimal meromorphic solution is formulated in 
Section~\ref{sec:min-sol}. In the same section, we find out that this 
solution satisfies one more complex Maryland equation with new  
parameters.  This is one of the  key observations leading  to the renormalization 
formula~\eqref{eq:main_renormalization_formulae}. The minimal solution 
is constructed in Section~\ref{sec:min-sol-constr}, where we obtain 
integral representations for it.

The above renormalization formula for the matrix product is as explicit as 
the renormalization formula obtained in~\cite{FK:12} for 
the Gaussian exponential sums.  These formulas have a similar structure.
In~\eqref{eq:main_renormalization_formulae},  $N_1\sim -\omega N$ for large $N$, 
and, as $0<\omega<1$, the analysis of the matrix product 
$\P_N(\omega,\theta,\eta,l)$ with a large number of factors is reduced to 
analysis of an analogous product with a smaller one. It is important to note 
that the factors $\Psi(\dots)$ in the right-hand side 
of~\eqref{eq:main_renormalization_formulae} have to be controlled only 
on the interval $[0,1)$. 

As in~\cite{FK:12}, one can easily show that, after a finite number 
(of order of $\log N$) of the renormalizations applied consequently to the 
matrix products  
$\P_N(\omega,\theta,\eta,l)$, $\P_{N_1}(\omega_1,\theta_1,\eta_1, l_1)$ 
etc., one can reduce the number of factors to one. In the analysis 
of the Gaussian sums, the main role was played by quasiclassical effects arising 
when the frequency $\omega_L=\left\{1/\omega_{L-1}\right\}$,  $L\ge 1$, 
$\omega_0=\omega$, is small. In the case of the Maryland equation,
there is an additional effect. It is well known that the product 
$\omega_0\omega_1\dots \omega_L$ exponentially decreases 
when $L$ grows. Therefore, after many renormalizations, one encounters 
the large parameter $l/(\omega_0\omega_1\dots\omega_L)$. So, one can expect
that the analysis of the behavior of $\P_N(\omega,\theta,\eta,l)$
for large $N$ can be very effective. We plan to employ this idea in our 
next publication.

Theorem~\ref{th:main} is obtained in Section~\ref{sec:renormalization}.
Its proof is based on monodromization method ideas. This method is a general 
renormalization approach suggested by V.~Buslaev and A.~Fedotov 
for studying difference equations on $\R$ with periodic coefficients. 
It was developed further  in papers of
A.~Fedotov and F.~Klopp, see the review article~\cite{F:13}. 
In Section~\ref{sec:renormalization}, we describe the monodromization idea 
and give a proof of a general renormalization formula for the case of 
difference equations on $\Z$  with coefficients being restrictions to $\Z$
of functions defined and periodic on $\R$. Note that a similar formula was 
stated without proof in~\cite{FK:08}. 
Formula~(\ref{eq:main_renormalization_formulae}) is  a 
corollary from the general one and from the observation that the 
minimal meromorphic solution of the complex Maryland  equation 
satisfies one more complex Maryland equation (with new parameters). 
This observation is equivalent to the fact that  {\it  the complex Maryland 
equation is invariant with respect to monodromization}. The reader 
finds more details in Section~\ref{sec:renormalization}.
\section{Minimal solutions}\label{sec:min-sol}
In this section, we discuss the difference equations on the complex plane 
only. Equation~(\ref{eq:complex-Mary}) is invariant with respect to 
multiplication by $e^{2\pi i z/\omega}$. Therefore,  if it has a meromorphic 
solution, it has meromorphic solutions growing as quickly as desired when
$\im z\to\pm \infty$. To define the minimal meromorphic solution,
i.e., the solution  having the slowest  growth for $\im z \to\pm \infty$,
one has to impose some natural conditions on the set of its poles.
To give the precise definition, we need to discuss  the set 
of solutions of~(\ref{eq:complex-Mary}).
\subsection{Solutions of  difference equations}
\label{sub:space}
Let us list well-known elementary properties of the solutions of the equation
\begin{equation}
   \label{eq:g-eq}
 \psi(z+\omega)+\psi(z-\omega)+v(z)\psi(z)=0, \quad z\in \C,  
\end{equation}
where $v$ is a given function, and $\omega>0$ is a given number.

Let $\psi$ and $\tilde\psi$ be two solutions  to~\eqref{eq:g-eq}.
It can be easily seen that the expression
\begin{equation}
  \label{eq:wron}
  w(\psi(z),\tilde \psi(z))= 
\psi(z)\tilde\psi(z-\omega)-\psi(z-\omega)\tilde\psi(z)
\end{equation}
is  $\omega$-periodic in $z$. It is called the Wronskian of 
$\psi$  and $\tilde\psi$.

If $ w(\psi(z),\tilde \psi(z))\ne0$ for all $z$, one can show that
any other solution  $\phi$ admits the representation
\begin{equation}
  \label{eq:bas}
  \phi(z)=a(z)\psi(z)+b(z)\tilde\psi(z),\quad z\in\C,
\end{equation}
with some $\omega$-periodic  $a$ and $b$.  This implies that the solution space 
of~\eqref{eq:g-eq} is a two-dimensional module over the ring of 
$\omega$-periodic functions. Note that~\eqref{eq:bas} and the Wronskian definition
imply that
\begin{equation}
  \label{eq:lin-comb-coef}
  a(z)=\frac{w(\phi(z),\tilde\psi(z))}{w(\psi(z),\tilde\psi(z))},\quad
  b(z)=\frac{w(\psi(z),\phi(z))}{w(\psi(z),\tilde\psi(z))}.
\end{equation}
\subsection{The simplest solutions to the complex Maryland equation
in a neighborhood of  $\pm i\infty$}\label{subsec:canonical-bases}
The periodicity of the potential in the  complex Maryland equation 
allows to consider   $+i\infty$ and  $-i\infty$ as two singular points. 
For  $Y\in\R$, we call the half-plane 
$\C_+(Y)=\{z\in\C\,:\im z>Y\}$ a neighborhood of
$+i\infty$, and we call  $\C_-(Y)=\{z\in\C\,:\im z<Y\}$ a neighborhood of 
$-i\infty$. 

The minimal meromorphic  solutions of the complex Maryland equation
are defined in terms of the solutions having the ``simplest'' behavior 
in neighborhoods of   $\pm i\infty$. 
The latter are described in
\begin{Th}\label{lm_bloch_existence} 
For sufficiently large $Y>0$, in $\C_+(Y)$, there exist analytic
solutions  $u_\pm$ to the complex Maryland equation such that
\begin{equation}
\label{eq_bloch_asymp}
u_\pm(z)=e^{\pm\frac{l-i\eta}{\omega}z} (1+o(1)),\quad \im z \to +i\infty,
\end{equation}
uniformly in  $z\in K_C=\{z\in\C\,:\, |\im z|\ge C |\re z|\}$,
where $C>0$ is an arbitrary fixed constant.
In the terminology of~\cite{BF:95},  $u_\pm$ are Bloch solutions,
i.e., $u_\pm(z+1)=\alpha_\pm(z)u_\pm(z)$ with some $\omega$-periodic 
factors $\alpha_\pm$.
\end{Th}
This theorem is proved in Section~\ref{subsec:bloch-sol-exist}. 

In a neighborhood of $-i\infty$, one can construct solutions $d_\pm$  
similar to $u_\pm$. It is convenient to define them by the formulas  
\begin{equation}\label{eq_bloch_down_from_up}
d_\pm(z)=\overline{u_\pm(\overline{z})}.
\end{equation}

Representations~(\ref{eq_bloch_asymp}) imply that
\begin{equation}
  \label{eq:bloch-sol-wron}
  w(u_+,u_-)=  e^{l-i\eta}-e^{-l+i\eta}+o(1),\quad \im z\to+\infty.
\end{equation}
Therefore, for sufficiently large  $Y$, the solutions $u_\pm$ form a basis for 
the space of the solutions defined on $\C_+(Y)$. Similarly,  $d_\pm$ 
form a basis for the space of the solutions defined on $\C_-(-Y)$.
We call the couples $(u_\pm)$ and $(d_\pm)$ the {\it canonical bases} for 
neighborhoods of  $+i \infty$ and $-i\infty$, respectively.

We define  $\alpha_\pm(z)=u_\pm(z+1)/u_\pm(z)$ and  
$\beta_\pm(z)=d_\pm(z+1)/d_\pm(z)$.
The functions  $\alpha_\pm$ and $\beta_\pm$ are $\omega$-periodic. 
Representations~(\ref{eq_bloch_asymp}) and~(\ref{eq_bloch_down_from_up}) 
imply that
\begin{gather}
  \label{u:floquet-factors}
  \alpha_\pm(z)=e^{\pm\frac{l-i\eta}{\omega}} (1+o(1)),\quad \im z \to +i\infty,\\
  \label{d:floquet-factors}
  \beta_\pm(z)=e^{\pm\frac{l+i\eta}{\omega}} (1+o(1)),\quad \im z \to -i\infty.
\end{gather}
\subsection{Minimal meromorphic solution of the complex 
Maryland equation} 
Let $\psi$ be a solution of~(\ref{eq:complex-Mary}) analytic in the strip
$S_0=\{z\in\C:\, |\re z|\le \omega\}$. 
\begin{Rem}\label{min-sol-analit} Equation~(\ref{eq:complex-Mary}) implies
that  $\psi$ can be continued to a  meromorphic function that can have 
poles only at the points $\pm(n+m\omega)$, $n,m\in\N$. Moreover,
the poles located at the points $\pm(1+\omega)$ are simple. For this new function,
we keep the old notation $\psi$.
\end{Rem}
Let $Y$ be chosen as in Theorem~\ref{lm_bloch_existence}. The solution 
$\psi$ admits the representations
\begin{gather}\label{psi-up}
 \psi(z)=A_+(z)\,u_+(z)+A_-(z)\,u_-(z),\quad z\in \C_+(Y),
\\
\label{psi-down}
\psi(z)=B_+(z)\,d_+(z)+B_-(z)\,d_-(z),\quad z\in C_-(Y),
\end{gather}
with some  $\omega$-periodic analytic coefficients $A_\pm$  and $B_\pm$.
\begin{Def}
The solution $\psi$ is called a {\it minimal meromorphic solution} 
to~\eqref{eq:complex-Mary} if the coefficients  $A_\pm$ and $B_\pm$ 
are bounded in $\C_\pm(Y)$.
\end{Def}
For a minimal meromorphic solution $\psi$, the limits  $a_{\pm}$  of 
$A_\pm$ for $\im z\to +\infty$ and the limits $b_{\pm}$ of $B_\pm$ for 
$\im z\to-\infty$ exist and are equal to the zeroth Fourier 
coefficients of  $A_\pm$ and $B_\pm$, respectively. We call  
$a_\pm$ and  $b_\pm$ the asymptotic coefficients of the minimal 
solution $\psi$.

In Section~\ref{sec:min-sol-constr}, we prove 
\begin{Th}\label{th:existence-light}
For $|\eta|<\pi(1+\omega)$, there exists a minimal meromorphic solution $\psi$
of the complex Maryland equation. It  is analytic in $\eta$, and its asymptotic 
coefficients do not vanish at $\eta\not\in \omega\Z$.
\end{Th}
\begin{Rem} {\sl The minimal meromorphic solution described in this theorem
can be continued to a meromorphic function of $\eta$ and $l$.}
\end{Rem}
The term ``minimal'' is explained by
\begin{Th}\label{th:min-light} Let $\psi$ be a minimal meromorphic solution, 
and let its asymptotic coefficients $a_\pm$ (or $b_\pm$) be non-zero.
Then
\point any other minimal solution coincides with $\psi$ up to a constant factor;
\point if $\phi$ is a minimal solution, and one of its asymptotic coefficients 
is zero, then $\phi\equiv 0$.
\end{Th}
When proving this theorem, we use
\begin{Le}\label{le:wron-psi-psi-plus-odin}
Let  $\psi$  be a minimal solution. Then
\begin{eqnarray*}
w(\psi(z+1),\psi(z))&=a_{+}a_{-}\,
\left(e^{\frac{l-i\eta}{\omega}}-e^{-\frac{l-i\eta}{\omega}}\right)
\left(e^{l-i\eta}-e^{-l+i\eta}\right)\\
&=b_{+}b_{-}\,\left(e^{\frac{l+i\eta}{\omega}}-e^{-\frac{l+i\eta}{\omega}}\right)
\left(e^{l+i\eta}-e^{-l-i\eta}\right).
\end{eqnarray*}
\end{Le}
\begin{proof} Remark~\ref{min-sol-analit} implies that the Wronskian  
of $\psi$  and $\psi(\,.\,+1)$ is analytic in the strip 
$\{z\in\C:\,-1<\re z<\omega\}$. As the Wronskian is  $\omega$-periodic, 
it is an entire function.

Let us study  the behavior of the Wronskian for $\im z\to +\infty$. 
Recall that  $u_\pm$ are Bloch solutions, $u_\pm(z+1)=\alpha_\pm(z) u_\pm(z)$,
where $\alpha_\pm$ are $\omega$-periodic. By means of~(\ref{psi-up}), 
we get
\begin{equation*}
  \psi(z+1)=A_+(z+1)\alpha_+(z)u_+(z)+A_-(z+1)\alpha_-(z)u_-(z).
\end{equation*}
From this formula, representation~(\ref{psi-up}) for $\psi$,
and the $\omega$-periodicity of  $\alpha_\pm$ and $A_\pm$, we deduce that
\begin{equation}\label{eq:aux1}
\begin{split}
 w(\psi(z+1),\,&\psi(z))=w(u_+(z),u_-(z))\,\cdot \\
&\cdot(A_{+}(z+1)A_{-}(z)\,\alpha_+(z)-A_{-}(z+1)A_{+}(z)\,\alpha_-(z)).
\end{split}
\end{equation}
Finally, the asymptotics~(\ref{eq:bloch-sol-wron}), the definition of 
the asymptotic coefficients, and~(\ref{u:floquet-factors}) imply that, 
as $\im z\to +\infty$, the right-hand side in~(\ref{eq:aux1}) tends to the 
first expression for the Wronskian given in 
Lemma~\ref{le:wron-psi-psi-plus-odin} .

One  similarly proves that, as $\im z\to-\infty$, \ $w(\psi(z+1),\,\psi(z))$ 
tends  to the second expression for the Wronskian described in this lemma.

As the Wronskian is an entire periodic function that has finite limits
as $\im z\to\pm\infty$, it is bounded. Being a bounded entire function, 
the Wronskian is independent of  $z$, and 
$$ w(\psi(z+1),\,\psi(z))=\lim_{\im z\to-\infty}w(\psi(z+1),\,\psi(z))=
\lim_{\im z\to+\infty}w(\psi(z+1),\,\psi(z)).$$ 
This leads to the statement of the lemma.
\end{proof}
\smallskip
Let us turn to the proof of Theorem~\ref{th:min-light}. 
\begin{proof}
Let $\phi$ be one more minimal meromorphic solution of the complex Maryland 
equation. Assume that the asymptotic coefficients  $a_\pm$ of $\psi$ 
are non-zero (the case of $b_\pm\ne 0$ is treated in the same way). The above 
lemma implies that the solutions $z\to\psi(z)$ 
and $z\to\psi(z+1)$ form a basis for the space of solutions of the complex 
Maryland equation. Therefore,  $\phi$ admits representation~(\ref{eq:bas}) with 
$\tilde \psi(z)=\psi(z+1)$. Recall that the coefficients in this 
representation are described in~(\ref{eq:lin-comb-coef}). 
As when proving Lemma~\ref{le:wron-psi-psi-plus-odin}, 
one shows that the Wronskians   $w(\phi(z),\tilde\psi(z))$ and 
$w(\psi(z),\phi(z))$ in~(\ref{eq:lin-comb-coef}) are independent 
of $z$. As $w(\psi(z+1),\psi(z))$ is also independent of $z$, see 
Lemma~\ref{le:wron-psi-psi-plus-odin}, to check the first 
statement of the theorem, it suffices to check that 
$w(\phi(z),\psi(z))$ vanishes at some $z$. Since $\psi$ and $\phi$ 
are analytic in  $\{z\in\C:\,|\re z|\le \omega\}$, the 
complex Maryland equation implies that both these solutions vanish 
at  $z=0$. Therefore, $\left.w(\phi(z),\psi(z))\right|_{z=0}=0$. 
This completes the proof of the first statement of the theorem.

Let us prove the second one. Lemma~\ref{le:wron-psi-psi-plus-odin} shows that
all the asymptotic coefficients of  $\psi$ are non-zero.
By the first statement, $\phi=C\,\psi$ with a constant  $C$. 
Therefore, asymptotic coefficients of the solutions $\phi$ and $\psi$
are proportional with the same constant $C$. As one of the asymptotic 
coefficients of $\phi$ is zero, we have $C=0$. Thus, $\phi=0$.
\end{proof}
\subsection{Second difference equation for the minimal solutions}
A central property of the minimal solutions is described by
\begin{Th}\label{th:second-eq}
Let $\psi$ be a minimal meromorphic solution of the complex Maryland 
equation. Assume that its asymptotic coefficients are non-zero. Then 
it solves the equation
\begin{gather}
  \label{eq:senond-eq}
  \psi(z+1)+\psi(z-1)+\lambda_1\cot(\pi z/\omega)\psi(z)=E_1\psi(z),\quad z\in\C,\\
\intertext{where}
\label{eq:lambda1-E1}
\lambda_1=-2\sin \eta_1{\rm sh}\, l_1,\quad E_1=2\cos \eta_1{\rm ch}\, l_1,
\end{gather}
and  $l_1$ and $\eta_1$ are related to  $l$ and $\eta$ by the formulas 
in~(\ref{eq_new_parameters_and_old}).
\end{Th}
\begin{Rem} As the minimal solution described in Theorem~\ref{th:existence-light}
is analytic in  $\eta$, it solves~\eqref{eq:senond-eq}  even if $\eta\in \omega\Z$, 
i.e., if  its asymptotic coefficients vanish.  
\end{Rem}
\begin{proof}
Lemma~\ref{le:wron-psi-psi-plus-odin} implies that $\psi$ and
$\tilde\psi=\psi(\,\cdot\,+1)$  form a basis for the solution space of
the complex Maryland equation. The function $\phi=\psi(\,\cdot\,-1)$ also
solves this equation and, therefore, admits representation~(\ref{eq:bas}) 
with the periodic coefficients described by~(\ref{eq:lin-comb-coef}). 
As $w(\psi(z+1),\psi(z))$ is independent of  $z$, see 
Lemma~\ref{le:wron-psi-psi-plus-odin},
the coefficient  $b$  in this representation is identically equal to  $-1$.
So, to prove the theorem, it suffices to calculate the coefficient
\begin{equation}\label{eq:a-expl}
a(z)=\frac{w(\psi(z-1),\psi(z+1))}{w(\psi(z),\psi(z+1))}.
\end{equation}
The Wronskian in the denominator in this formula is described by 
Lemma~\ref{le:wron-psi-psi-plus-odin}. Let us discuss the Wronskian 
in the numerator. Remark~\ref{min-sol-analit} implies that  
$z\to w(\psi(z-1),\psi(z+1))$ is a meromorphic $\omega$-periodic 
function analytic in the strip  $0<\re z<\omega$ and that, on the 
boundary of the strip, it may have poles only at the points $z=0$ 
and $z=\omega$.  By a reasoning similar to one from the proof
of Lemma~\ref{le:wron-psi-psi-plus-odin}, one shows that, as
$\im z\to +\infty$, the Wronskian tends to 
\begin{equation*}
a_{+}a_{-}
\left(e^{-\frac{2(l-i\eta)}{\omega}}-e^{\frac{2(l-i\eta)}{\omega}}\right)
\left(e^{l-i\eta}-e^{-l+i\eta}\right),\\
\end{equation*}
and,  as $\im z\to-\infty$, it tends to
\begin{equation*}
b_{+}b_{-}
\left(e^{-\frac{2(l+i\eta)}{\omega}}-e^{\frac{2(l+i\eta)}{\omega}}\right)
\left(e^{l+i\eta}-e^{-l-i\eta}\right).
\end{equation*}
We see that $a$ is a meromorphic  $\omega$-periodic function, it 
may have poles only at $z\in \omega\Z$, these poles are simple, 
and $a$ tends to constants as $\im z\to\pm \infty$. This implies that 
$a(z)=E_1-\lambda_1\cot(\pi z/\omega)$ with some constants $E_1$ 
and $\lambda_1$. 

The above asymptotics for   $w(\psi(z-1),\psi(z+1))$ and the asymptotics
for $w(\psi(z+1),\psi(z))$ in Lemma~\ref{le:wron-psi-psi-plus-odin} 
imply that
\begin{equation*}
a(z)\to \begin{cases}
e^{\frac{l-i\eta}{\omega}}+e^{-\frac{l-i\eta}{\omega}}, & \im z\to +\infty,\\
e^{\frac{l+i\eta}{\omega}}+e^{-\frac{l+i\eta}{\omega}}, & \im z\to-\infty.
\end{cases}
\end{equation*}
Therefore, $E_1=2\cos(\eta/\omega){\rm ch}\,(l/\omega)$ and
$\lambda_1=-2\sin(\eta/\omega){\rm sh}\,(l/\omega)$.
This completes the proof of the theorem.
\end{proof}
\section{Monodromization and renormalization formulas}
\label{sec:renormalization}
In this section, first, following~\cite{F:13}, we recall basic ideas
of the monodromization theory, next, we prove a general renormalization 
formula for matrix cocycles, then, we  describe some corollaries from 
these constructions for the Maryland equation and, after that, 
prove Theorem~\ref{th:main}.
\subsection{Monodromization}
\subsubsection{Monodromy matrix}
Consider the matrix solutions of the equation
\begin{equation}
    \label{eq:matrix-line}
    \Psi\,(x+\omega)=M\,(x)\,\Psi\,(x),\qquad x\in\R, 
  \end{equation}
where $M$ is a given $1$-periodic ${\rm SL}\,(2,\C)$-valued function 
and  $0<\omega<1$ is a fixed number.

For any solution $\Psi$  of equation~(\ref{eq:matrix-line}), 
$\det\Psi$ is an $\omega$-periodic function. We call a solution $\Psi$ 
{\it fundamental}, if  $\det\Psi$ is independent of $x$ and does not vanish.
Below, we assume that $\Psi$ is a fundamental solution.

A function $\tilde\Psi:\R\to {\rm GL}(2,\C)$ solves~(\ref{eq:matrix-line})  
if and only if  
\begin{equation}
   \label{eq:general-matrix-sol}
   \tilde\Psi\,(x)=\Psi\,(x)\cdot p\,(x),\quad \forall\,x\in\R,
\end{equation}
where $p$ is an  $\omega$-periodic matrix function.

The function $x\to \Psi(x+1)$ is a solution of~(\ref{eq:matrix-line}) 
together with  $\Psi$. Therefore,
\begin{equation*}
\Psi\,(x+1)=\Psi\,(x)\cdot p\,(x),\quad p(x+\omega)=p(x),\qquad \forall x\in\R.
\end{equation*}
The matrix 
\begin{equation*}
M_1(x)=p^t(\omega x),
\end{equation*}
where ${}^t$ denotes transposition, is the  {\it monodromy} matrix 
corresponding to the fundamental solution $\Psi$. Like the matrix  $M$
from the input equation, the monodromy matrix is $\omega$-periodic and
unimodular.
\subsubsection{Very short introduction to the monodromization theory }
Let $\omega_1$ be the Gauss transform of  $\omega$,  i.e.
$\omega_1=\left\{\frac{1}{\omega}\right\}$. Consider the equation
\begin{equation}\label{M1}
\Psi_1(x+\omega_1)=M_1(x)\,\Psi_1(x), \quad x\in\R,
\end{equation}
where  $M_1$ is the monodromy matrix corresponding to a fundamental
solution  $\Psi$ of~\eqref{eq:matrix-line}. 
We say that~\eqref{M1} is a {\it monodromy} equation obtained 
from~\eqref{eq:matrix-line} via {\it monodromization}.

The monodromy equation~\eqref{M1} is similar to the input one:
the matrices $M$ and $M_1$ are both unimodular and $1$-periodic. 
Therefore, the {\it monodromization procedure} can be continued:
one can consider the monodromy matrix corresponding to a fundamental 
solution of~\eqref{M1} and the corresponding monodromy equation 
and so on. In result, one arrives to an infinite sequence 
of difference equations similar to the input one. There are deep 
relationships between these equations (see, for example, 
Theorem~\ref{th_renormalization_formula}). The leading idea of 
the monodromization method is to analyze  solutions of the 
input equation by analyzing properties of the dynamical system
that defines the coefficients of each  equation  
in the sequence in terms of the coefficients of the previous one. 
\subsubsection{Renormalization of matrix cocycles}
Together with~(\ref{eq:matrix-line}), consider the family of 
difference equations on $\Z$  
\begin{equation}
  \label{eq:matrix-Z}
  \Psi_{k+1}=M(\omega k+\theta)\,\Psi_k, \quad k\in\Z,
\end{equation}
where $0\le\theta<1$ is the parameter indexing the equations. Let 
$k\to P_k(M,\omega,\theta)$ be the solution of~(\ref{eq:matrix-Z}) 
equal to the identity matrix when $k=0$. It is obvious that
$P_k(M,\omega,\theta)=M(\omega(k-1)+\theta)\dots M(\omega+\theta)M(\theta)$,
when  $k>0$, and $P_k(M,\omega,\theta)=M^{-1}(\omega k+\theta)\dots 
M^{-1}(\omega-2\theta)M(\omega-\theta)$, when   $k<0$.
\begin{Th}[on renormalizations of matrix cocycles]
\label{th_renormalization_formula}
Let $\Psi$ be a fundamental solution of~(\ref{eq:matrix-line}), 
and let $M_1$ be the corresponding monodromy matrix.
Then, for all $N\in \Z$,
\begin{gather}
\label{eq_main_renormalization_formulae}
P_N(M,\omega,\theta)=\Psi(\{\theta+N\omega\})\sigma_2
P_{N_1}(M_1,\omega_1,\theta_1)\sigma_2\Psi^{-1}(\theta),\\
\label{eq:new-cocycle-parameters}
N_1=-[\theta+N\omega],\qquad\omega_1=\left\{1/\omega\right\},\qquad 
\theta_1=\left\{\theta/\omega\right\},
\end{gather}
where  $\sigma_2$ is the matrix defined in~(\ref{eq:Psi-sigma}).
\end{Th}
A renormalization formula similar to~\eqref{eq_main_renormalization_formulae}
was  stated without proof in~\cite{FK:08}. 
Formula~\eqref{eq_main_renormalization_formulae} relates the solution
$P(M,\omega,\theta)$ of equation~(\ref{eq:matrix-Z}) to the solution 
$P(M_1,\omega_1,\theta_1)$ of the equation of the same form 
but with the matrix $M_1$ and the parameters $\omega_1$ and $\theta_1$ 
instead of $M$, $\omega$ and $\theta$. 
\begin{proof} 
In the case of $N=0$, the statement is obvious.
Assume that $N>0$ (the case of $N<0$ is treated similarly).
Equation~\eqref{eq:matrix-line} implies that
$$\Psi(\theta+N\omega)=M(\theta+(N-1)\omega)\,M(\theta+(N-2)\omega)...
M(\theta)\Psi(\theta)= P_N(M,\omega,\theta)\Psi(\theta).$$
The solution $\Psi$ being fundamental, the matrix $\Psi(\theta)$ is invertible. 
Therefore,
\begin{gather}
\label{eq_cocycle_repr_fundam_sol}
P_N(M,\omega,\theta)=\Psi(\theta+N\omega)\Psi^{-1}(\theta).
\end{gather}
The definition of the monodromy matrix $M_1$ implies that
$\Psi(x)=\Psi(x-1)M_1^t\left(\frac{x-1}{\omega}\right)$.
Using this relation to express $\Psi(\theta+N\omega)$ in terms of 
$\Psi(\{\theta+N\omega\})$, we get
\begin{equation*}
\begin{split}
P_N(M,\omega,\theta)=\Psi(\{\theta+N\omega\})\,&
M_1^t\left(\frac{\theta+N\omega-[\theta+N\omega]}{\omega}\right)\,\dots\\
&\dots\, M_1^t\left(\frac{\theta+N\omega-2}{\omega}\right)\,
M_1^t\left(\frac{\theta+N\omega-1}{\omega}\right)\,\Psi^{-1}(\theta).
\end{split}
\end{equation*}
Taking into account the $1$-periodicity of  $M_1$ and  
using~\eqref{eq:new-cocycle-parameters}, we arrive at the formula
$$
P_N(M,\omega,\theta)=\Psi(\{\theta+N\omega\})\,M_1^t(\theta_1+N_1\omega_1)\,
\dots\,M_1^t(\theta_1-2\omega_1)\,M_1^t(\theta_1-\omega_1)\,\Psi^{-1}(\theta).
$$
For any $A\in SL(2,\C)$, one has $A^t=\sigma_2 A^{-1}\sigma_2$.
Therefore,
\begin{equation*}
\begin{split}
\Psi(\{\theta+N\omega\})\,\sigma_2\,M_1^{-1}(\theta_1+N_1\omega_1)\,&\dots\, 
M_1^{-1}(\theta_1-2\omega_1)M_1^{-1}(\theta_1-\omega_1)\,\sigma_2\,
\Psi^{-1}(\theta)=\\
&=\Psi(\{\theta+N\omega\})\sigma_2
P_{N_1}(M_1,\omega_1,\theta_1)\sigma_2^{-1}\Psi^{-1}(\theta).
\end{split}
\end{equation*} 
This implies the statement of the theorem.
\end{proof}
\subsection{Monodromization and  the Maryland equation}
\subsubsection{Monodromy matrix for the complex Maryland equation}
Let $\psi$ be the minimal meromorphic solution of the complex Maryland equation
described in Theorem~\ref{th:existence-light}.  In terms of $\psi$, we construct
the matrix  $\Psi$ as in~(\ref{eq:Psi-sigma}). One has
\begin{Le}\label{le:Mary:fund-sol} 
The function $z\to\Psi(z,\eta,l)$ solves the equation
\begin{equation}
  \label{eq:Mary-R}
  \Psi(z+\omega)=\F(z,\eta,l)\Psi(z),\quad z\in\C.
\end{equation}
If $\eta\not\in\omega\Z$, the solution $\Psi$ is fundamental.
\end{Le}
\begin{proof} The first statement is obvious. As $\det\Psi$ 
equals the Wronskian of  $\psi$ and $\psi(\,\cdot\,-1)$, the 
second statement follows from Lemma~\ref{le:wron-psi-psi-plus-odin}. 
\end{proof}
\begin{Rem} {\sl By Theorem~\ref{th:existence-light}, the solution
$\Psi$ analytically depends on $\eta$.}
\end{Rem}
The definition of the monodromy matrix and Theorem~\ref{th:second-eq} 
imply
\begin{Th}\label{pro:Mary-M} 
If $\eta\not\in\omega\Z$, the monodromy matrix corresponding 
to the fundamental solution $\Psi(\cdot,\eta,l)$
equals $\F(\cdot,\eta_1,l_1)$, where  $\eta_1$ and $l_1$ are 
defined by the formulas in ~(\ref{eq_new_parameters_and_old}).
\end{Th}
\subsubsection{Invariance with respect to monodromization}
Theorem~\ref{pro:Mary-M} means that the {\it matrix Maryland 
equation}~(\ref{eq:Mary-R}) is invariant with respect to monodromization:  
after  monodromization, it appears to be  transformed into 
the matrix Maryland equation with new parameters. 
Like~(\ref{eq:Mary-R}), the latter is equivalent to a complex Maryland 
equation~(\ref{eq:complex-Mary}) \ (the equation with new parameters), 
and one can say that the complex Maryland equation is invariant 
with respect to monodromization.   Actually, it is this invariance 
that leads to the renormalization 
formula~\eqref{eq:main_renormalization_formulae}.
 
In~\cite{BF:01}, the authors consider difference equations
on $\C$ the coefficients of which are trigonometric polynomials, and describe
equation families invariant with respect to monodromization. 
One of these families contains an equation related to the famous Almost Mathieu 
equation (in the same way as the Maryland equation is related to the complex 
Maryland equation). However, in general case, investigation of  
the trigonometric polynomial coefficients transformation, which occurs in 
result of monodromization, appears to be a very non-trivial problem. Only 
for very special trigonometric polynomials, this transformation is known to 
be elementary ~\cite{FK:08}.  The explicit transformation of the complex 
Maryland equation coefficients proved  in this paper looks to be a very 
rare phenomenon. 
\subsubsection{Renormalization formula for the Maryland equation}
For $\eta\not\in\omega\Z$, Theorem~\ref{th:main} immediately follows 
Theorems~\ref{th_renormalization_formula} and~\ref{pro:Mary-M}. As the 
renormalization formula~(\ref{eq:main_renormalization_formulae}) is an 
equality of two functions analytic in $\eta$,  the statement of 
Theorem~\ref{th:main} remains valid also for $\eta\in\omega\Z$.
\section{Construction of the minimal meromorphic solution}
\label{sec:min-sol-constr}
In this section, we construct a minimal meromorphic solution of the 
complex Maryland equation~(\ref{eq:complex-Mary}). First, we describe 
a solution analytic in $\C\setminus \R$. Next, we check that it can 
be continued to a meromorphic function.  Then, we compute the 
asymptotics of this function for  $\im  z\to\pm\infty$ and show that 
it is minimal. Finally, we prove Theorem~\ref{lm_bloch_existence}.

Below,  $C$ denotes different positive constants (independent of  $z$),
and $K_C=\{z\in\C\,:\, |\im z|\ge C |\re z|\}$.
\subsection{Solution analytic in $\C\setminus\R$}
We begin with a short description of a special function we use 
in this section. 
\subsubsection{$\sigma$-function}\label{sect_sigma_function}
An analogous function was introduced and systematically studied 
in the diffraction theory~\cite{BLG}. Later it appeared in other 
domains, e.g.~\cite{FKV} and~\cite{BF:01}. Below, we rely on  the 
last paper.

The special function  $\sigma$ can be uniquely defined as 
a meromorphic solution of the difference equation
\begin{equation}
\label{eq_sigma_main_equation}
\sigma(z+\pi\omega)=(1+e^{-iz})\,\sigma(z-\pi\omega)
\end{equation}
which is analytic in the strip $S=\{z\in\C\,:\,|\re z|<\pi(1+\omega)\}$, does not
vanish there and admits in  $S$  the following uniform asymptotics:
\begin{gather}
\label{eq_sigma_down_asymp}
\sigma(z)=1+o(1),\quad \im z\to-\infty,\\
\label{eq_sigma_up_asymp}
\sigma(z)=e^{-\frac{i z^2}{4\pi\omega}+\frac{i\pi}{12 \omega}+\frac{i\pi\omega}{12}}
(1+o(1)), \quad \im z\to \infty.
\end{gather}
The  asymptotics~(\ref{eq_sigma_down_asymp}) and~(\ref{eq_sigma_up_asymp}) 
appear to be uniform in $K_C$  for any fixed  $C$.
The poles of  $\sigma$ are located at the points
\begin{equation}
\label{eq_sigma_poles} z=-\big(\pi(1+\omega)+2\pi\omega k+2\pi m\big), \quad
k,m\in\N\cup\{0\},
\end{equation}
and its zeros are described by the formulas
\begin{equation}
\label{eq_sigma_zeroes} 
z=\pi(1+\omega)+2\pi\omega k+2\pi m,\quad
k,m\in\N\cup\{0\};
\end{equation}
the zero at  $z=\pi(1+\omega)$ and the pole at  $z=\pi(1+\omega)$ are simple.
We note that 
$$\res_{z=\pi(1+\omega)}\frac{1}{\sigma(z)}=
-\sqrt{\omega}\,e^{\frac{i\pi}{12\omega}+\frac{i\pi\omega}{12}+\frac{i\pi}{4}}.$$
The $\sigma$ function solves one more difference equation
\begin{equation}
\label{eq_sigma_second_equation}
\sigma(z+\pi)=(1+e^{-iz\,/\,\omega})\sigma(z-\pi)
\end{equation}
and satisfies the relations
\begin{equation}
\label{eq:sigma-funct-prop}
\sigma(z)=e^{-\frac{i}{4\pi\omega}z^2+\frac{i\pi}{12\omega}+
\frac{i\pi\omega}{12}}\,/\,\sigma(-z)\quad\text{and}\quad
\overline{\sigma(\overline{z})}=1\,/\,\sigma(-z).
\end{equation}
\subsubsection{Analytic solution of the complex Maryland equation} 
Here, we always assume that $\im z\ne 0$.
Instead of $E$ and $\lambda$, we use the parameters $\eta\in\R$ and $l>0$, 
see~(\ref{eq_Gamma_l_E_lambda}). It is convinient to consider $|\eta|<\pi+\omega$.
We construct a  solution represented by  a contour integral,
and we begin by describing the contour. 

Put
\begin{align*}
\C'=\C\setminus\big((\eta-il-\pi(1+\omega)-\R_+)
\cup(-\eta-il+\pi(1+\omega)+\R_+)\cup\\
\cup(-\eta+il-\pi(1+\omega)-\R_+)\cup (\eta+il+\pi(1+\omega)+\R_+)\big).
\end{align*}
For $z\in\C$, denote by  $D(z)$ the set of rays  going in $\C$ to infinity 
in parallel to the  vectors corresponding to the complex numbers
\begin{equation}
\label{eq_gamma_asymptotes}
\tau=e^{i\alpha},\quad \alpha\in(-\arg z,-\arg z+\pi);\quad -\pi<\arg z<\pi.
\end{equation}
We assume that  $\gamma=\gamma(z)$ is a curve in $\C'$ and that, first, 
it goes from $-i\infty$ to $-2il$ along a ray of $D(z)$, then, it goes from
$-2il$ to $2il$ along $i\R$ and, finally, it goes to  $+i\infty$ 
along one more ray of  $D(z)$.
\begin{Pro} \label{lm_Y_z_not_in_R}
If $|\eta|<\pi(1+\omega)$, the formula
\begin{equation}
\label{eq_Y_integral_representation}
\Y(z)=\sin (\pi z) \sin\left(\frac{\pi}{\omega}z\right)
\int_{\gamma(z)}e^{\frac{i p z}{\omega}}\frac{\sigma(p+\eta-i l)
\sigma(p-\eta+il)}{\sigma(p-\eta-i l)\sigma(p+\eta+il)}\,dp
\end{equation}
defines a solution of~(\ref{eq:complex-Mary})  analytic in  
$z\in\C\setminus\R$. This solution is also analytic 
in $\eta$.
\end{Pro}
\begin{proof}
The description of the  poles and zeros of the $\sigma$-function
implies the analyticity of the integrand in $p\in\C'$. 
The convergence and  analyticity of the integral 
in~\eqref{eq_Y_integral_representation} follow from 
estimates~(\ref{eq_sigma_down_asymp}) and~(\ref{eq_sigma_up_asymp}) 
and from the definition of the curve $\gamma(z)$. Let us check that  $\Y$ 
solves~(\ref{eq:complex-Mary}). Denote the contour integral 
by $X(z)$, and denote the integrand by $e^{\frac{i p z}{\omega}} \hat X(p)$.
Equation~(\ref{eq:complex-Mary})  for $\Y$ is equivalent to the equation
\begin{equation}
\label{eq_main_X_equation}
\begin{split}
\sin(\pi(z+\omega))X(z+\omega)&+\sin(\pi(z-\omega))X(z-\omega)+\\
 &+2\left(\cos\eta\,\,{\rm ch}\, l\,\sin(\pi z) +\sin\eta\,\,{\rm sh}\,l\, 
\cos (\pi z)\right)X(z)=0.
\end{split}
\end{equation}
Assume that  $\gamma=\gamma(z)$ goes to $\pm i\infty$  along rays from the set
$D(z-\omega)\cap D(z+\omega)$. Then this curve can be used as the integration
contour in the representations  for each of  the functions $X(z)$, 
$X(z-\omega)$ and $X(z+\omega)$. This allows to 
transform~\eqref{eq_main_X_equation} to the equation
\begin{equation}\label{last-equality}
\begin{split}\int_{\gamma+\pi\omega} &e^{\frac{i (p+\omega) z}{\omega}}
(1+e^{-i(p+\eta-i l)})(1+e^{-i(p-\eta+i l)})\hat{X}(p-\pi\omega)\,dp-\\
&-\int_{\gamma-\pi\omega} e^{\frac{i (p+\omega) z}{\omega}}
(1+e^{-i(p-\eta-i l)})(1+e^{-i(p+\eta+i l)})\hat{X}(p+\pi\omega)\,dp=0.
\end{split}
\end{equation} 
Equation~(\ref{eq_sigma_main_equation})  and the definition of  $\hat X$ 
imply that
\begin{equation*}
\hat{X}(p+\pi\omega)=\frac{(1+e^{-i(p+\eta-i l)})(1+e^{-i(p-\eta+i l)})}
{(1+e^{-i(p-\eta-i l)})(1+e^{-i(p+\eta+i l)})}\hat{X}(p-\pi\omega).
\end{equation*}
So, it suffices to check that, in~\eqref{last-equality}, one can replace 
$\gamma\pm \pi\omega$ by $\gamma$. Consider the first integral 
in~\eqref{last-equality}, the second one can be treated similarly. 
Translate $\gamma+ \pi\omega$ to $\gamma$ along the real line. 
Asymptotics~\eqref{eq_sigma_down_asymp} and~\eqref{eq_sigma_up_asymp}
imply that, translating the integration contour, 
we do not break the convergence of the integral. So, we need only to 
check that, when being translated,  the integration contour does not 
cross any poles of the integrand. The description of the zeros and poles of 
$\sigma$ shows that the contour can cross only the poles of 
$\hat X(\,.-\pi\omega)$ located at  $p=\pm (\eta-il)-\pi$. These poles being 
simple, the expression $(1+e^{-i(p+\eta-i l)})(1+e^{-i(p-\eta+i l)})
\hat{X}(p-\pi\omega)$ has no singularities at  $p=\pm (\eta-il)-\pi$.
This implies the desired. 
\end{proof}
\begin{Rem} {\sl Using equation~(\ref{eq_sigma_second_equation}), 
one can directly check that  $\Y$ solves~(\ref{eq:senond-eq}). By means 
of~(\ref{eq:sigma-funct-prop}), one proves that
$\Y(\overline{z})=-\overline{\Y(z)}$.} 
\end{Rem}
\subsection{Real $z$}\label{subsec:regularisation}
Here, we prove
\begin{Pro}\label{lm:Y_continuation_real}
The solution $\Y$ can be continued to a meromorphic function that may have poles 
only at  $z=\pm(\omega k+m)$, \ $k,m\in\N$.
\end{Pro}
\begin{proof}
If $\im z\ne0$,~(\ref{eq_Y_integral_representation}) implies that
\begin{equation*}
\begin{split}
\Y(z)&=-\frac{1}{4}\int_{\gamma(z)}
(e^{\frac{i (p+\pi\omega+\pi) z}{\omega}}-e^{\frac{i (p+\pi\omega-\pi) z}{\omega}}
-e^{\frac{i (p-\pi\omega+\pi) z}{\omega}}+e^{\frac{i (p-\pi\omega-\pi) z}{\omega}})
\hat{X}(p)\,dp=\\
&=-\frac{1}{4}\sum_{s_1,s_2=\pm 1}
s_1 s_2\int_{\gamma(z)+s_1\pi\omega+s_2\pi}e^{\frac{i p z}{\omega}}
\hat{X}(p-s_1\pi\omega-s_2\pi)\,dp\\
&=-\frac{1}{4}\sum_{s_1,s_2=\pm 1}
s_1 s_2\int_{\gamma(z)}e^{\frac{i p z}{\omega}}
\hat{X}(p-s_1\pi\omega-s_2\pi)\,dp+2\pi i R(z),
\end{split}
\end{equation*}
where $e^{i p z/\omega}\hat X$ is the integrand 
in~(\ref{eq_Y_integral_representation}), and $R(z)$ denotes the sum 
of the residues appeared when deforming the integration contour. 
The function $R$ is entire. To analyse it, we note that, 
in~(\ref{eq_Y_integral_representation}), both the  integrand 
and the part of the integration countur situated in $\{|\im p|\le 2l\}$ 
are independent of $z$. Furthermore, the poles of the integrand 
are located on the lines $\im p=\pm l$. This implies that $R$ is given 
by one and the same formula both for  $\im z>0$ and for  $\im z<0$.  
By means of~(\ref{eq_sigma_main_equation}) 
and~(\ref{eq_sigma_second_equation}), the last representation for  
$\Y$ can be transformed to the form
\begin{equation*}
\Y(z)=\int_{\gamma(z)}e^{\frac{i p z}{\omega}}
\frac{A\hat{X}(p-\pi-\pi\omega)}
{(\cos p-\cos (\eta+il))(\cos(p/\omega)-\cos ((\eta+il)/\omega))}
\,dp+2\pi i R(z),
\end{equation*}
where $A={\rm sh}\, l\,{\rm sh}\, (l/\omega)\,\sin\eta\,\sin(\eta/\omega)$.

Using~(\ref{eq_sigma_down_asymp}) and~(\ref{eq_sigma_up_asymp}), 
one can easily see that, for any fixed $C>0$, in $K_C$, the integrand 
admits the estimates
$O(e^{i(z-1-\omega)p/\omega})$ for $p \to -i\infty$ and  
$O(e^{i(z+1+\omega)p/\omega})$ for $p \to +i\infty$.
 
Assume that  $|\re z|<1+\omega$. Thanks to the last two estimates, 
we can deform the integration contour to the imaginary axis (both 
for $\im z>0$ and $\im z<0$). In the strip $|\re z|<1+\omega$, the 
obtained contour integral converges for all $z$  and defines an 
analytic function.  Therefore, $\Y$ is analytic in the strip
$\{z\in\C\,:\,|\re z|<1+\omega\}$. It can be continued to a meromorphic 
one directly via equation~(\ref{eq:complex-Mary}). This equation also  
implies the statement on the poles of  $\Y$.
\end{proof}
\subsection{Behavior of  $\Y$ for $ \im z\to\pm\infty$}
\label{subsect_asymptotics}
Here, first, we get the asymptotics of $\Y$ for $\im z\to\pm\infty$, and then,
we check that  $\Y$ is a minimal meromorphic solution of~(\ref{eq:complex-Mary}).
\begin{Pro}\label{lm_Y_z_asymptotic behavior}
Fix $C>0$. If $|\eta|<\pi(1+\omega)$, then, in $K_C$, 
\begin{gather}
\label{eq_Y_up_expansion}
\Y(z)=e^{(l-i\eta) z/\omega}(a_+ +o(1)) +e^{-(l-i\eta) z/\omega}(a_- +o(1)),
\quad \im z\to +\infty,\\
\label{eq_Y_down_expansion}
\Y(z)=e^{(l+i\eta)z/\omega} (b_+ +o(1))+e^{-(l+i\eta)z/\omega} (b_- +o(1)),
\quad \im z\to -\infty,
\end{gather}
where
\begin{gather}\label{eq:Y_i_infty_leading_coefs}
a_\pm=\frac{\pi i}{2}
\frac{\sigma(\pi(1+\omega)\mp 2\eta)\sigma(\pi(1+\omega)\mp 2il)}
{\sigma(\pi(1+\omega)\mp 2(\eta+il))}\,\res_{p=\pi(1+\omega)}\frac{1}{\sigma(p)},\\
\label{eq_Y_asymp_coef_below_and_under}
b_\pm(\eta,l)=-\overline{a_\pm(\eta,l)}.
\end{gather}
\end{Pro}
\begin{Rem} {\sl Formulas~(\ref{eq:Y_i_infty_leading_coefs}) and the description 
of the zeros of the $\sigma$-function imply that $a_-=b_-=0$ at
$\eta=0,\omega,2\omega\dots$}
\end{Rem}  
\begin{proof}
Assume that  $z\in\C_+\cap K_C$. As $z\not\in\R$, we 
use~(\ref{eq_Y_integral_representation}). For sufficiently small
$\delta>0$, for all $z\in\C_+\cap K_C$, \ $D(z)$  contains rays 
parallel to the vectors $e^{\pm i\delta}$. Therefore, for all  $z\in\C_+\cap K_C$, 
in~(\ref{eq_Y_integral_representation}), we can choose one and the same 
integration contour  $\gamma$. 

When being translated to the right along the real line, the integration contour 
can cross poles of the integrand. These are zeros  of the denominator 
in~(\ref{eq_Y_integral_representation}) located at the points 
$\pm (il+\eta)+\pi(1+\omega)+2\pi(n+\omega m)$, \ $n,m=0,1,2\dots$. 
Let us translate the contour $\gamma$ to  $\gamma+c$, where  $c>0$ 
is chosen so that, in the course of  translation, the contour crosses
the poles at  $\pm (il+\eta)+\pi(1+\omega)$ and that, after the translation,
it does not contain any pole of the integrand.

The  function $\Y$ equals the sum of the term $I(z)$ containing the integral along 
$\gamma+c$  and $S(z)$, the sum  of (a finite number of) the contributions of the 
residues appeared when deforming $\gamma$ to $\gamma+c$. 
One has
\begin{equation}\label{S-formula}
  S(z)= e^{(l-i\eta) z/\omega}(a_+ +o(1)) +e^{-(l-i\eta) z/\omega}(a_- +o(1)),
\quad \im z\to +\infty,
\end{equation}
with $a_\pm$ given by~\eqref{eq:Y_i_infty_leading_coefs}. 
When deriving the last formula,  we take into account the fact that 
the zero of $\sigma$-function at  $p=\pi(1+\omega)$ is simple. 
Let us estimate $I(z)$. Along $\gamma+c$,  $|\hat X(p)|$ is bounded.
Therefore,
\begin{equation*}
\begin{split}
  |I(z)|\le \Const\, &e^{\pi (1+\omega) |\im z|/\omega} \int_{\gamma+c}|e^{i p z/\omega}\,dp|\\
&=\Const\,e^{\pi (1+\omega) |\im z|/\omega-c\,\im z}\, \int_{\gamma}|e^{i p z/\omega}\,dp|.
\end{split}
\end{equation*}
The last integral converges. When $\im z$ increases, it increases
exponentially.  Therefore, if $c$ is sufficiently large, then, as
$\im z\to+\infty$, the $I$ integral becomes small with respect to 
both exponentials in~\eqref{S-formula}.  This 
implies~(\ref{eq_Y_up_expansion}).

The proof of~(\ref{eq_Y_down_expansion}) is similar to the proof 
of~(\ref{eq_Y_up_expansion}), but one  translates $\gamma$ to the left. 
Omitting the details, we note only that, to  
get~(\ref{eq_Y_asymp_coef_below_and_under}), one has to 
use~(\ref{eq:sigma-funct-prop}).
\end{proof}
Now, one easily checks the main statement of the section:
\begin{Th} \label{th:existence-full}
The solution $\Y$ is a minimal meromorphic solution 
to~(\ref{eq:complex-Mary}); its asymptotic coefficients 
are given in~(\ref{eq:Y_i_infty_leading_coefs}) 
and~(\ref{eq_Y_asymp_coef_below_and_under}).
\end{Th}
Theorem~\ref{th:existence-light} is an immediate corollary 
of this theorem. 
\begin{proof}
By Proposition~\ref{lm:Y_continuation_real},  $\Y$ is analytic in 
$|\re z|\le \pi\omega$. Consider the coefficients  $A_\pm$ 
in formula~(\ref{psi-up}) representing  $\psi=\Y$ as a linear combination 
of the canonical basis solutions $u_\pm$.  In view of 
Section~\ref{sub:space},  $A_\pm(z)=\pm w(\Y(z),u_\mp(z))\,/w(u_+(z),u_-(z))$.
Assume that  $z\in\C_+$ is in the $\omega$-neighborhood of the line $i(l+i\eta)\R$.
Then  $e^{\pm(l-i\eta)z/\omega}$ are of order of one, and 
using~(\ref{eq_Y_up_expansion}),~(\ref{eq_bloch_asymp}) 
and~(\ref{eq_bloch_down_from_up}), we get
$A_\pm(z)=a_\pm+o(1)$ as  $\im z\to+\infty$. This representation is uniform 
in $\re z$.  As $A_\pm$ are $\omega$-periodic, these representations remain  
valid and uniform in $\C_+$. One studies the coefficients $B_\pm$ in the 
representation~(\ref{psi-down}) for $\psi=\Y$ similarly. This leads to  the 
statement of the theorem.
\end{proof}
\subsection{Construction of the canonical Bloch solutions}
\label{subsec:bloch-sol-exist}
Here, using techniques developed in~\cite{BF:01} and~\cite{F-K:05b},
we prove Theorem~\ref{lm_bloch_existence}. The proof is carried out in 
several steps:
\par\noindent{\bf 1.} \ Consider equation~(\ref{eq:Mary-R}) equivalent 
to~(\ref{eq:complex-Mary}). As $\im z\to+\infty$, in this equation, 
the matrix takes the form 
$\F(z,\eta,l)=\begin{pmatrix}2\cos(\eta+il) & -1\\1 & 0
\end{pmatrix}+O(e^{-2\pi\im z})$.
The eigenvalues of the leading term equal   $\nu^{\pm 1}$, \  
$\nu=e^{-i(\eta+il)}$. Put $\phi=V^{-1}\psi$, where 
$V=\begin{pmatrix}1 & 1\\ 1/\nu & \nu\end{pmatrix}$, and $\psi$ is 
a vector solution of~(\ref{eq:Mary-R}). In a neighborhood of $+i\infty$,
$\phi$ solves the equation
\begin{equation}\label{eq:diagonalized}
  \phi(z+h)=(D+m(z))\phi(z),\quad D=\begin{pmatrix}\nu & 0\\ 0 & 1/\nu
\end{pmatrix}, \quad m(z)=O(e^{-2\pi\im z}).
\end{equation}
\par\noindent{\bf 2.} \ Let $\phi_{1}(z)$ and $\phi_{2}(z)$ be the first 
and the second components of the vector $\phi(z)$.
Put $\Phi(z)=\phi_2(z)/\phi_1(z)$. Then
\begin{equation}\label{eq_phi_difference_equation}
\Phi(z+\omega)=\frac{(1/\nu+m_{22}(z))\Phi(z)+m_{21}(z)}
{\nu +m_{11}(z)+m_{12}(z)\Phi(z)}.
\end{equation}
We construct a solution of this equation by means of the 
technique described in Section 4.1.1 from~\cite{F-K:05b}. 
Consider the sequence of functions defined by the formulas 
\begin{equation*}
\Phi_{n+1}(z+\omega)=\frac{(1/\nu+m_{22}(z))\Phi_n(z)+m_{21}(z)}
{\nu +m_{11}(z)+m_{12}(z)\Phi_n(z)},\quad n\ge 0,\quad \Phi_0(z)\equiv 0.
\end{equation*}
Let $D\subset\C$ be a domain.
Repeating the proof of Proposition 4.1 from~\cite{F-K:05b}, 
we show that if $|\nu|>1$,  and  $\sup_{z\in D}|m(z)|$  is sufficiently 
small, then, for all $n\in\N$ and  $z\in D$, \  $|\Phi_n(z)|\le 1$,
the  sequence $\{\Phi_n\}$ converges uniformly  in $z\in D$, and
the limit $\Phi$ solves~\eqref{eq_phi_difference_equation}.
As, in a neighborhood of  $+i\infty$,  $m$ is analytic and $1$-periodic 
and satisfies the estimate in~(\ref{eq:diagonalized}), we conclude that,  
in a neighborhood of $+i\infty$, there exists an analytic $1$-periodic 
bounded solution $\Phi$ of equation~\eqref{eq_phi_difference_equation}.
As $\Phi$ is $1$-periodic and bounded, it can be represented by the 
Fourier series of the form $\Phi(z)=\sum_{m=0}^\infty q_n e^{2\pi i m z}$.
Substituting it into~\eqref{eq_phi_difference_equation}, one checks that
$\Phi(z)=O(e^{-2\pi\im z})$ as $\im z\to+\infty$.
\par\noindent{\bf 3.} \ If $\Phi$ solves~\eqref{eq_phi_difference_equation}
and $\phi_1$ satisfies the equation
\begin{equation}\label{eq:first-component}
  \phi_1(z+\omega)=(\nu +m_{11}(z)+m_{12}(z)\Phi(z))\phi_1(z),
\end{equation}
then the vector with the components  $\phi_1(z)$ and $\phi_2=\phi_1(z)\Phi(z)$
solves~(\ref{eq:diagonalized}). 

\par\noindent{\bf 4.} \ Let $\Phi$ be the function constructed in the 
second step. To construct a solution of~\eqref{eq:first-component}, 
we use Lemma 2.3 from~\cite{BF:01}. It can be formulated in the following 
way:
\begin{Le}\label{Le-g}  Let $g$ be a  $1$-periodic function analytic in 
a neighborhood of $+i\infty$ such that $g(i\infty)=0$. Then equation 
$f(z+\omega)-f(z)=g(z)$ 
with a fixed $0<\omega<1$ has a solution  analytic in a neighborhood of  
$+i\infty$ and decreasing  as $\im z\to+\infty$ uniformly in 
$\{z\in\C\,:\,|\re z|\le C\}$, where $C>0$ is an arbitrary fixed constant.
\end{Le}
Define $A(z)=(\nu +m_{11}(z)+m_{12}(z)\Phi(z))$.
The estimates for $\Phi$ and $m$ for $\im z\to +\infty$ imply that
$A(z)=\nu(1+o(1))$. Choose the branch of $B=\ln A$ so that
$B=-i(\eta+il)+g$ and  $g(z)=o(1)$ as $\im z\to+\infty$.
Let $f$ be the function constructed by means of Lemma~\ref{Le-g} 
in terms of  $g$. Then  $\phi_1(z)=e^{-i(\eta+il)z/\omega+f(z)}$ 
solves~\eqref{eq:first-component}. Using the observation made at the third
step, one constructs in terms of $\phi_1$ a solution of~(\ref{eq:diagonalized}) 
analytic in a neighborhood of $+i\infty$ and such that, as $\im z\to+\infty$,
$\phi(z)=e^{-i(\eta+il)z/\omega}\left(\begin{pmatrix} 1\\0\end{pmatrix}
+o(1)\right)$  uniformly in $\{z\in\C\,:\,|\re z|\le C\}$, \ 
$C>0$ being a fixed constant.
\par\noindent{\bf 5.} \ The function $\phi$ is a Bloch solution, i.e.,
$\phi(z+1)=\alpha(z)\phi(z)$, where $\alpha$ is an 
$\omega$-periodic. Indeed, a direct calculation shows that, for any 
solution $f$ of the equation $f(z+\omega)-f(z)=g(z)$ with a given $1$-periodic 
function $g$, \  $f(z+1)-f(z)$ is $\omega$-periodic. 
This implies that $\alpha(z)=\phi_1(z+1)/\phi_1(z)$ is $\omega$-periodic. As
$\Phi$ is $1$-periodic, one has   $\phi_2(z+1)/\phi_2(z)=\phi_1(z+1)/\phi_1(z)$.
This implies the needed. 
\par\noindent{\bf 6.} \  Fix $C_1>0$. Let us show that the asymptotics of 
$\phi$ is uniform in  $K_{C_1}$. Consider the coefficient $\alpha$ 
from the definition of the Bloch solution $\phi$. As it is  $\omega$-periodic,  
the asymptotics for $\phi$ in $\{z\in\C\,:\,|\im z|\le C\}$ implies that, 
as $\im z\to +\infty$, \ $\alpha(z)=e^{-i(\eta+il)/\omega+O(e^{-2\pi\im z/\omega})}$ 
uniformly in $\re z$. Assume that $z\in K_C$. Let $N$ be the integer part of
$\re z$. One has
\begin{equation*}
 \phi(z)=\left(\prod_{n=1}^{N}\alpha(z-n)\right) \phi(z-N)=
e^{-i(\eta+il)N/\omega+O(\im z\,e^{-2\pi\im z/\omega})}\,\phi(z-N). 
\end{equation*}
Substituting in this formula the asymptotics for $\phi$ justified for 
bounded $|\re z|$, we obtain the needed.
\par\noindent{\bf 7.} \ One can easily see that the first component $\psi_1$ 
of a vector solution $\psi$ of~(\ref{eq:Mary-R}) 
satisfies~(\ref{eq:complex-Mary}). Let  $\psi=V\,\phi$, where 
$\phi$ the solution of~(\ref{eq_phi_difference_equation})
constructed in the previous steps. By the result of the first step, 
$\psi$ solves~(\ref{eq:Mary-R}). We construct solutions $u_\pm$ 
of~(\ref{eq:complex-Mary}) by the formulas
$u_+(z)=\psi_1(z)$ and $u_-(z)=\overline{u_+(-\overline{z})}$.
One can easily check that these solutions have all the properties listed 
in Theorem~\ref{lm_bloch_existence}. We omit the elementary calculations.
\bibliographystyle{abbrv}

\begin{thebibliography}{10}
%
\bibitem{BLG}
V.~Babich, M.~Lyalinov and V.~Grikurov. 
\newblock {\em Diffraction theory: the Sommerfeld-Malyuzhinets
technique.}
\newblock Oxford, UK: Alpha Science, 2008.
%
\bibitem{BF:95}
V.~Buslaev, A.~Fedotov.
\newblock Bloch solutions of difference equations.
\newblock {\em St.Petersburg Math. J.} 1996, 7(4): 561-594, 1996.
%
\bibitem{BF:01}
V.~Buslaev, A.~Fedotov.
\newblock On the difference equations with periodic coefficients.
\newblock {\em Advances in Theor. and Math. Phys.} 5(6):1105-1168, 2001.
%
\bibitem{CFKS}
H.L. Cycon, R.G. Froese, W. Kirsch, B. Simon.
\newblock Schr\"{o}odinger operators, with application
to quantum mechanics and global geometry.
\newblock Berlin, Springer-Verlag, 1987.
%
\bibitem{FKV}
L. Faddeev, R. Kashaev and A. Volkov.
\newblock Strongly coupled quantum discrete Liouville theory. I: 
Algebraic approach and duality.
\newblock {\em Commun. Math. Phys.} 219:199-219, 2001.
%
\bibitem{F:13}
A.~Fedotov.
\newblock Monodromization method in the theory of almost periodic equations.
\newblock {\em Algebra i Analys [In Russian]}   25(2):203-236, 2013.
\newblock To appear in English in {\em St.Petersburg Math. J.}
%
\bibitem{F-K:05b}
A.~Fedotov and F.~Klopp. 
\newblock Strong resonant tunneling, level repulsion and spectral 
type for one-dimensional adiabatic quasi-periodic Schr\"odinger operators. 
\newblock {\em Annales Scientifiques de l'Ecole Normale Sup\'erieure, 
4e s\'erie}, 38(6):889-950, 2005.
%
\bibitem{FK:08}
A.~Fedotov and F.~Klopp.
\newblock Pointwise Existence of the Lyapunov Exponent for a 
Quasiperiodic Equation.
\newblock {\em Mathematical results in quantum mechanics.} Eds: Ingrid Beltita,
Gheorghe Nenciu \& Radu Purice. World Sci. Pub., 2008,  55-66.
%
\bibitem{FK:12}
A.~Fedotov, F.~Klopp 
\newblock An exact renormalization formula for Gaussian exponential sums 
and applications.
\newblock  {\em American Journal of Mathematics} 134(3): 711-748, 2012.
%
\bibitem{PF}
A.~Figotin, L.~Pastur, 
\newblock Spectra of random and almost periodic operators.
\newblock Springer-Verlag, 1991.
%
\bibitem{FGP}
S.~Fishman, D.R.~Grempel, R.E.~Prange
\newblock Wave functions at a mobility edge: An example of a singular
continuous spectrum.
\newblock {\em Phys. Rev. B,} 28(12):7370-7372, 1983.
%
\end{thebibliography}

\end{document}